\documentclass[3p,12pt,singlecolumn]{elsarticle}

\usepackage{url}
\usepackage[utf8]{inputenc}
\usepackage{amsmath}
\usepackage{amssymb}
\usepackage{amsthm}
\usepackage{stmaryrd}
\usepackage{tikz}
\usetikzlibrary{automata,positioning}

\newcommand{\fivespaces}{\;\;\;\;\;}
\newcommand{\tenspaces}{\fivespaces\fivespaces}
\newcommand{\twentyspaces}{\tenspaces\tenspaces}
\newcommand{\mylabel}[1]{\, \mathbf{(#1)}}
\newcommand{\Lp}{\stackrel{\mbox{\tiny{PEG}}}{\leadsto}}
\newcommand{\Lpp}{\stackrel{\mbox{\tiny{PEGF}}}{\leadsto}}
\newcommand{\Lpl}{\stackrel{\mbox{\tiny{PEGL}}}{\leadsto}}
\newcommand{\throw}{\Uparrow}
\newcommand{\Tup}[2]{(#1,\,#2)}

\newcommand{\Any}{X}
\newcommand{\Ex}{v?}
\newcommand{\Ey}{w?}
\newcommand{\Suf}{{\tt smallest}}
\newcommand{\Suff}[2]{\Suf(#1,\,#2)}
\newcommand{\Fail}{{\tt fail}}
\newcommand{\Nil}{{\tt nil}}
\newcommand{\List}{L}
\newcommand{\Mon}[1]{\{#1\}}
\newcommand{\J}{{\tt join}}
\newcommand{\Jj}[2]{\J(#1,\,#2)}

\newcommand{\Jv}{{\tt joinVar}}

\def\drawplusplus#1#2#3{\hbox to 0pt{\hbox to #1{\hfill\vrule height #3 depth
      0pt width #2\hfill\vrule height #3 depth 0pt width #2\hfill
      }}\vbox to #3{\vfill\hrule height #2 depth 0pt width
      #1 \vfill}}
\def\concat{\mathrel{\drawplusplus {12pt}{0.4pt}{8pt}}}

\newtheorem{proposition}{Proposition}[section]
\newtheorem{lemma}[proposition]{Lemma}
\newcommand{\Peg}[2]{#1[#2]}
\newcommand{\Pgg}[1]{\Peg{G}{#1}}
\newcommand{\Mat}[2]{#1\;\,#2\,}
\newcommand{\Matg}[2]{\Mat{\Pgg{#1}}{#2}}

\title{Error Reporting in Parsing Expression Grammars}

\journal{Science of Computer Programming}

\begin{document}

\begin{frontmatter}

\author{André Murbach Maidl}
\ead{andre.murbach@pucpr.br}
\address{Polytechnic School -- PUCPR --
Curitiba -- Brazil}

\author{Fabio Mascarenhas}
\ead{mascarenhas@ufrj.br}
\address{Department of Computer Science -- UFRJ --
Rio de Janeiro -- Brazil}

\author{Sérgio Medeiros}
\ead{sergiomedeiros@ect.ufrn.br}
\address{School of Science and Technology -- UFRN --
	Natal -- Brazil}

\author{Roberto Ierusalimschy}
\ead{roberto@inf.puc-rio.br}
\address{Department of Computer Science -- PUC-Rio --
Rio de Janeiro -- Brazil}

\begin{abstract}
Parsing Expression Grammars (PEGs) describe top-down parsers.
Unfortunately, the error-reporting techniques used in conventional top-down parsers do not directly apply to parsers based on Parsing Expression Grammars (PEGs), so they have to be somehow simulated.
While the PEG formalism has no account of semantic actions, actual
PEG implementations add them, and we show how to simulate an
error-reporting heuristic through these semantic actions.  

We also propose a complementary error reporting strategy
that may lead to better error messages: labeled failures.
This approach is inspired by exception handling of
programming languages, and lets a PEG define different
kinds of failure, with each ordered choice operator
specifying which kinds it catches. Labeled failures
give a way to annotate grammars for better
error reporting, to express some of the error
reporting strategies used by deterministic parser combinators,
and to encode predictive top-down parsing in a PEG.
\end{abstract}

\begin{keyword}
parsing \sep
error reporting \sep
parsing expression grammars \sep
packrat parsing \sep
parser combinators
\end{keyword}

\end{frontmatter}

\section{Introduction} \label{sec:intro}

When a parser receives an erroneous input,
it should indicate the existence of syntax errors.
Errors can be handled in various ways. The easiest is just to report that an error was found, where is was found and what was expected at that point and then abort. At the other end of the spectrum we find mechanisms that attempt to parse the complete input, and report as many errors as best as  possible.

The LL($k$) and LR($k$) methods detect syntax errors very efficiently
because they have the \emph{viable prefix} property, that is,
these methods detect a syntax error as soon as $k$ tokens are read
and cannot be used to extend the thus far accepted part of the input
into a viable prefix of the language
\cite{aho2006cpt}. LL($k$) and LR($k$) parsers can use this
property to produce suitable, though generic, error
messages.

Parsing Expression Grammars (PEGs)~\cite{ford2004peg} are a
formalism for describing the syntax of programming languages.
We can view a PEG as a formal description of a top-down parser
for the language it describes. PEGs have a concrete syntax based on
the syntax of {\em regexes}, or extended regular expressions.
Unlike Context-Free Grammars (CFGs),
PEGs avoid ambiguities in the definition of the grammar's
language due to the use of an {\em ordered choice} operator.

More specifically, a PEG can be interpreted as a the specification
of a recursive
descent parser with restricted (or local) backtracking.
This means that the alternatives of a choice are tried in
order; as soon as an alternative recognizes an input prefix,
no other alternative of this choice will be tried, but when an
alternative fails to recognize an input prefix, the parser
backtracks to try the next alternative.

On the one hand, PEGs can be interpreted as a formalization of a
specific class of top-down parsers~\cite{ford2004peg}; on the other hand,
PEGs cannot use error handling techniques that are often applied
to predictive top-down parsers, because these techniques assume the parser
reads the input without backtracking~\cite{ford2002packrat}.
In top-down parsers without backtracking, it is possible to
signal a syntax error as soon as the next input symbol cannot be accepted.
In PEGs, it is more complicated to identify the cause
of an error and the position where it occurs, because failures
during parsing are not necessarily errors, but just an
indication that the parser cannot proceed and a different
choice should be made \emph{elsewhere}.

Ford~\cite{ford2002packrat} has already identified this
limitation of error reporting in PEGs, and, in his
parser generators for PEGs, included an heuristic for
better error reporting. This heuristic simulates the error
reporting technique that is implemented in top-down parsers
without backtracking. The idea is to track the position
in the input where the farthest failure occurred, as
well as what the parser was expecting
at that point, and report this to the user in case of errors.

Tracking the farthest failure position and context
gives us PEGs that produce error messages
similar to the automatically produced error
messages of other top-down parsers; they tell
the user the position where the error was encountered,
what was found in the input at that position,
and what the parser was expecting to find.

In this paper, we show how grammar writers can use
this error reporting technique even in PEG implementations
that do not implement it, by making use of semantic
actions that expose the current position in the input
and the possibility to access some form of mutable state
associated with the parsing process. 

We also propose a complementary approach for error
reporting in PEGs, based on the concept of \emph{labeled
failures}, inspired by the standard exception handling
mechanisms as found in programming languages. Instead of just
failing, a labeled PEG can produce different kinds of failure
labels using a {\em throw} operator. Each label can be
tied to a more specific error message. PEGs can also
{\em catch} such labeled failures, via a change to
the ordered choice operator. We formalize labeled
failures as an extension of the semantics of regular
PEGs.

With labeled PEGs we can express some alternative error
reporting techniques for top-down parsers with local
backtracking. We can also encode
predictive parsing in a PEG, and we show
how to do that for $LL(*)$ parsing, a powerful predictive
parsing strategy.

The rest of this paper is organized as follows:
Section~\ref{sec:pegs} contextualizes the problem of error
handling in PEGs, explains in detail the failure tracking
heuristic, and shows how it can be realized in PEG implementations
that do not support it directly; Section~\ref{sec:rel}
discusses related work on error reporting for top-down
parsers with backtracking;
Section~\ref{sec:lf}
introduces and formalizes the concept of labeled failures,
and shows how to use it for error reporting;
Section~\ref{sec:labelsfft} compares the error messages
generated by a parser based on the failure tracking heuristic  
with the ones generated by a parser based on labeled failures;
Section~\ref{sec:labelsrelated} shows how labeled failures
can encode some of the techniques of Section~\ref{sec:rel},
as well as predictive parsing;
finally, Section~\ref{sec:conc} gives some concluding
remarks.

\section{Handling Syntax Errors with PEGs} 
\label{sec:pegs}

\begin{figure}[t]
\begin{align*}
\it{Tiny} & \leftarrow \it{CmdSeq}\\
\it{CmdSeq} & \leftarrow \it{(Cmd \; {\tt SEMICOLON}) \;
  (Cmd \; {\tt SEMICOLON})*}\\
\it{Cmd} & \leftarrow \it{IfCmd \; / \; RepeatCmd \; / \;
  AssignCmd \; / \; ReadCmd \; / \; WriteCmd}\\
\it{IfCmd} & \leftarrow \it{{\tt IF} \; Exp \; {\tt THEN} \; CmdSeq
  \; ({\tt ELSE} \; CmdSeq \; / \; \varepsilon) \; {\tt END}}\\
\it{RepeatCmd} & \leftarrow \it{{\tt REPEAT} \; CmdSeq \;
  {\tt UNTIL} \; Exp}\\
\it{AssignCmd} & \leftarrow \it{{\tt NAME} \; {\tt ASSIGNMENT} \; Exp}\\
\it{ReadCmd} & \leftarrow \it{{\tt READ} \; {\tt NAME}}\\
\it{WriteCmd} & \leftarrow \it{{\tt WRITE} \; Exp}\\
\it{Exp} & \leftarrow \it{SimpleExp \;
  (({\tt LESS} \; / \; {\tt EQUAL}) \; SimpleExp \; / \; \varepsilon})\\
\it{SimpleExp} & \leftarrow \it{Term \;
  (({\tt ADD} \; / \; {\tt SUB}) \; Term)*}\\
\it{Term} & \leftarrow \it{Factor \;
  (({\tt MUL} \; / \; {\tt DIV}) \; Factor)*}\\
\it{Factor} & \leftarrow \it{{\tt OPENPAR} \; Exp \; {\tt CLOSEPAR}
  \; / \; {\tt NUMBER} \; / \; {\tt NAME}}
\end{align*}
\caption{A PEG for the Tiny language}
\label{fig:tiny}
\end{figure}

In this section, we use examples to present in more detail how a
PEG behaves badly in the presence of syntax errors.
After that, we present a heuristic proposed by Ford
\cite{ford2002packrat} to implement error reporting in PEGs.
Rather than using the original notation and semantics of PEGs 
given by Ford~\cite{ford2004peg}, our examples use the
equivalent and more concise notation and semantics proposed
by Medeiros et al.~\cite{medeiros2011re2peg,medeiros2012left,mascarenhas2014}.
We will extend both the notation and the semantics in Section~\ref{sec:lf}
to present PEGs with labeled failures.

A PEG $G$ is a tuple $(V,T,P,p_{S})$ where $V$ is a finite set of
non-terminals, $T$ is a finite set of terminals, $P$ is a total
function from non-terminals to \emph{parsing expressions} and $p_{S}$
is the initial parsing expression.
We describe the function $P$ as a set of rules of the form
$A \leftarrow p$, where $A \in V$ and $p$ is a parsing
expression.
A parsing expression, when applied to an input string, either
fails or consumes a prefix of the input and returns the
remaining suffix. The abstract syntax of parsing expressions is
given as follows, where $a$ is a terminal, $A$ is a non-terminal,
and $p$, $p_1$ and $p_2$ are parsing expressions:
\[\it{
p = \varepsilon \; | \; a \; | \; A \; | \; p_1 p_2 \; | \;
  p_1 / p_2 \; | \; p* \; | \; !p
}\]

Intuitively,
$\varepsilon$ successfully matches the empty string, not changing
the input;
$a$ matches and consumes itself or fails otherwise;
$A$ tries to match the expression $P(A)$;
$p_1 p_2$ tries to match $p_1$ followed by $p_2$;
$p_1 / p_2$ tries to match $p_1$;
if $p_1$ fails, then it tries to match $p_2$;
$p*$ repeatedly matches $p$ until $p$ fails, that is, it
consumes as much as it can from the input;
the matching of $!p$ succeeds if the input does not match $p$
and fails when the the input matches $p$, not consuming any input in either case;
we call it the negative predicate or the lookahead predicate.

Figure \ref{fig:tiny} presents a PEG for the Tiny language~\cite{louden1997ccp}.
Tiny is a simple programming language with a syntax that resembles
Pascal's. 
We will use this PEG, which can be seen as the equivalent of an LL(1) CFG,
to show how error reporting differs between top-down parsers without backtracking
and PEGs.

PEGs usually express the language syntax at the character level,
without the need of a separate lexer.
For instance, we can write the lexical rule \texttt{IF} as follows,
assuming we have non-terminals {\tt Skip}, which consumes whitespace,
and {\tt IDRest}, which consumes any character that may
be present on a proper suffix of an identifier\footnote{In the presented PEG, we omitted the lexical rules for brevity.}:
\[
  {\tt IF} \leftarrow \it{if \; !IDRest \; Skip}
\]

Now, we present an example of erroneous Tiny code so we can compare
approaches for error reporting.
The program in Figure~\ref{fig:tinyerror} is missing a semicolon
(\texttt{;}) in the assignment in line \texttt{5}.
A predictive top-down parser that aborts on the first error presents
an error message like:

\begin{verbatim}
    factorial.tiny:6:1: syntax error, unexpected 'until', expecting ';'
\end{verbatim}

The error is reported in line \texttt{6} because the parser cannot
complete a valid prefix of the language, since it unexpectedly
finds the token \texttt{until} when it was expecting a command
terminator (\texttt{;}).

\begin{figure}[t]
\begin{verbatim}
 01  n := 5;
 02  f := 1;
 03  repeat
 04    f := f * n;
 05    n := n - 1
 06  until (n < 1);
 07  write f;
\end{verbatim}
\caption{Program for the Tiny Language with a Syntax Error}
\label{fig:tinyerror}
\end{figure}

In PEGs, we can try to report errors using the remaining suffix, but
this approach usually does not help the PEG produce an error
message like the one shown above.
In general, when a PEG finishes parsing the input, a remaining non-empty suffix means that parsing did not reach the end of file due to a syntax error. However, this suffix usually does not indicate the
actual place of the error, as the error will have caused the PEG
to backtrack to another place in the input.

In our example, the problem happens when the PEG tries to recognize
the sequence of commands inside the \texttt{repeat} command.
Even though the program has a missing semicolon (\texttt{;}) in the
assignment in line \texttt{5}, making the PEG fail to recognize
the sequence of commands inside the \texttt{repeat} command, this
failure is not treated as an error.
Instead, this failure makes the recognition of the \texttt{repeat}
command also fail.
For this reason, the PEG backtracks the input to line \texttt{3}
to try to parse other alternatives for {\em CmdSeq},
and since these do not exist, its ancestor {\em Cmd}.
Since it is not possible to recognize a command other than
\texttt{repeat} at line \texttt{3},
the parsing finishes without consuming all the input.
Hence, if the PEG uses the remaining suffix to produce an error
message, the PEG reports line 3 instead of line 6 as the location 
where no further progress can be made.

There is no
perfect method to identify which information is the most relevant to
report an error. In our example it is easy for the parser to
correctly report what the error is, but it is easy to construct examples
where this is not the case. If we add the semicolon
in the end of line 6 and remove line 3, a predictive top-down parser
would complain about finding an {\tt until} where it expected 
another statement, while the actual error is a missing {\tt repeat}.

According to Ford \cite{ford2002packrat}, using the information of the farthest position
that the PEG reached in the input is a heuristic that provides
good results.
PEGs define top-down parsers and try to recognize
the input from left to right, so the position farthest to
the right in the input that a PEG reaches during parsing usually is
close to the real error \cite{ford2002packrat}. The same idea
for error reporting in top-down parsings with backtracking was
also mentioned in Section 16.2 of~\cite{grune2010ptp}. 

Ford used this heuristic to add error reporting to his PEG implementation using packrat parsers~\cite{ford2002packrat}.
A packrat parser generated by Pappy \cite{ford2002pappy}, Ford's PEG
parser generator, tracks the farthest position and uses this position
to report an error.
In other words, this heuristic helps packrat parsers to simulate the
error reporting technique that is implemented in deterministic parsers.

Although Ford only has discussed his heuristic in relation to packrat
parsers, we can use the farthest position heuristic to add error reporting
to any implementation of PEGs that provides semantic actions.
The idea is to annotate the grammar with semantic actions that track
this position. While this seems onerous, we just need to add annotations
to all the lexical rules to implement error reporting.

For instance, in Leg \cite{leg}, a PEG parser generator with
Yacc-style semantic actions, we can annotate the rule for \texttt{SEMICOLON} as follows, where {\tt |} is Leg's ordered choice operator, and following it is a semantic action (in the notation used
by Leg): 
\begin{verbatim}
    SEMICOLON = ";" Skip | &{ updateffp() }
\end{verbatim}

The function \texttt{updateffp} that the semantic action calls updates
the farthest failure position in a global variable if the current
parsing position is greater than the position that is stored in
this global, then makes the whole action fail so parsing continues
as if the original failure had occurred.

However, storing just the farthest failure position does not give
the parser all the information it needs to produce an informative
error message.
That is, the parser has the information about the position where
the error happened, but it lacks the information about what terminals
failed at that position.
Thus, we extend our approach by including the terminals in the annotations
so the parser can also track these names in order to compute the set of
expected terminals at a certain position:

\begin{verbatim}
    SEMICOLON = ";" Skip | &{ updateffp(";") }
\end{verbatim}

The extended implementation of \texttt{updateffp} keeps,
for a given failure position, the names of all the symbols
expected there.
If the current position is greater than the farthest failure,
{\tt updateffp} initializes this set with just the given name.
If the current position equals the farthest failure, {\tt updateffp}
adds this name to the set.

Parsers generated by Pappy also track the set of expected terminals,
but with limitations.
The error messages include only symbols and keywords that were defined
in the grammar as literal strings.
That is, the error messages do not include terminals that were defined
through character classes.

Our approach of naming terminals in the semantic actions avoids the
kind of limitation found in Pappy, though it increases the annotation
burden because the implementor of the PEG is also responsible for
adding one semantic action for each terminal and its respective name.

The annotation burden can be lessened in implementations of PEGs that
treat parsing expressions as first-class objects, because this makes
it possible to define functions that annotate the lexical parts of the grammar to track errors, record information about the expected
 terminals to produce good error messages, and enforce lexical
  conventions such as the presence of surrounding whitespace.
For instance, in LPEG \cite{lpeg,ierusalimschy2009lpeg}, a PEG
implementation for Lua that defines patterns as first-class objects,
we can annotate the rule \textit{CmdSeq} as follows,
where the patterns \verb'V"A"', \verb'p1 * p2', and \verb'p^0'
are respectively equivalent to parsing expressions
$A$, $p_1 p_2$, and $p*$ (in the notation used by LPEG):

\begin{verbatim}
    CmdSeq = V"Cmd" * symb(";") * (V"Cmd" * symb(";"))^0;
\end{verbatim}

The function \texttt{symb} receives a string as its only argument and returns a parser that is equivalent to the parsing expression that we used in the Leg example.
That is, \verb'symb(";")' is equivalent to
\verb'";" Skip | &{ updateffp(";") }'.

We implemented error tracking and reporting using semantic actions
as a set of parsing combinators on top of LPeg and used these
combinators to implement a PEG parser for Tiny.
It produces the following error message for the example we have
been using in this section:

\begin{verbatim}
    factorial.tiny:6:1: syntax error, unexpected 'until',
                        expecting ';', '=', '<', '-', '+', '/', '*'
\end{verbatim}

We tested this PEG parser for Tiny with other erroneous inputs and in all
cases the parser identified an error in the same place as a top-down
parser without backtracking.
In addition, the parser for Tiny produced error messages that are
similar to the error messages produced by packrat parsers generated
by Pappy.
We annotated other grammars and successfully obtained
similar results.
However, the error messages are still generic; they are not
as specific as the error messages of a hand-written top-down
parser.

\section{Error Reporting in Top-Down Parsers with Backtracking} 
\label{sec:rel}

In this section, we discuss alternative approaches for error reporting
in top-down parsers with backtracking.

Mizushima et al. \cite{mizushima2010php} proposed a cut operator
($\uparrow$) to reduce the space consumption of packrat parsers;
the authors claimed that the cut operator can also be used to
implement error reporting in packrat parsers, but the authors did
not give any details on how the cut operator could be used for this
purpose.
The cut operator is borrowed from Prolog to annotate pieces of
a PEG where backtracking should be avoided.
PEGs' ordered choice works in a similar way to Prolog's green cuts,
that is, they limit backtracking to discard unnecessary solutions.
The cut proposed for PEGs is a way to implement Prolog's white cuts,
that is, they prevent backtracking to rules that will certainly fail.

The semantics of cut is similar to the semantics of an
\texttt{if-then-else} control structure and
can be simulated through predicates.
For instance, the PEG (with cut) $A \leftarrow B \uparrow C / D$ is
functionally equivalent to the PEG (without cut)
$A \leftarrow B C / !B D$ that is also functionally equivalent to
the rule $A \leftarrow B[C,D]$ on Generalized Top-Down Parsing
Language (GTDPL), one of the parsing techniques that influenced the
creation of PEGs \cite{ford2002packrat,ford2002pappy,ford2004peg}.
On the three cases, the expression $D$ is tried only if the expression
$B$ fails.
Nevertheless, this translated PEG still backtracks 
whenever $B$ successfully matches and $C$ fails.
Thus, it is not trivial to use this translation to implement error
reporting in PEGs. 

\emph{Rats!}~\cite{grimm2006rats} is a popular packrat parser
that implements error reporting with a strategy similar to
Ford's, with the change that it always reports error positions
at the start of productions, and pretty-prints non-terminal names
in the error message. For example, an error in a {\em ReturnStatement}
non-terminal becomes {\tt return statement expected}.

Even though error handling is an important task for parsers,
we did not find any other research results about error handling in PEGs,
beyond the heuristic proposed by Ford and the cut operator
proposed by Mizushima et al.
However, parser combinators \cite{hutton1992hfp} present some
similarities with PEGs so we will briefly discuss them for the
rest of this section.

In functional programming it is common to implement recursive
descent parsers using parser combinators \cite{hutton1992hfp}.
A parser is a function that we use to model symbols of the
grammar.
A parser combinator is a higher-order function that we use to
implement grammar constructions such as sequencing and choice.
One kind of parser combinator implements parsers that
return a list of all possible results of a parse, effectively
implementing a recursive descent parser with full backtracking.
Despite being actually deterministic in behavior (parsing the
same input always yields the same list of results), these
combinators are called {\em non-deterministic parser 
combinators} due to their use of a non-deterministic choice operator.
We get parser combinators that have the same semantics as PEGs
by changing the return type from list of results to \texttt{Maybe}.
That is, we use {\em deterministic parser combinators} that return
\texttt{Maybe} to implement recursive descent parsers with limited
backtracking~\cite{spivey2012maybe}. 
In the rest of this paper, whenever we refer to parser combinators
we intend to refer to these parser combinators with limited backtracking.

Like PEGs, most deterministic parser combinator libraries 
also use ordered choice, and thus suffer from the same
problems as PEGs with erroneous inputs, where the point
that the parser reached in the input is ususally far away
from the point of the error.

Hutton \cite{hutton1992hfp} introduced the \texttt{nofail} combinator
to implement error reporting in a quite simple way:
we just need to distinguish between failure and error during parsing.
More specifically, we can use the \texttt{nofail} combinator to
annotate the grammar's terminals and non-terminals that should
not fail; when they fail, the failure should be transformed into an
error. The difference between an error and a failure is that
an ordered choice just propagates an error in its first alternative
instead of backtracking and trying its second alternative, so any
error aborts the whole parser.
This technique is also called the \emph{three-values} technique~\cite{partridge1996fv}
because the parser finishes with one of the following values:
\texttt{OK}, \texttt{Fail} or \texttt{Error}.

Röjemo \cite{rojemo1995epc} presented a \texttt{cut} combinator
that we can also use to annotate the grammar pieces where
parsing should be aborted on failure, on behalf of efficiency and
error reporting.
The \texttt{cut} combinator is different from the cut operator\footnote{Throughout this paper, we refer to {\em combinators}
of parser combinators and to {\em operators} of PEGs, but these terms are effectively interchangeable.} ($\uparrow$) for PEGs because the combinator is abortive and unary
while the operator is not abortive and nullary.
The \texttt{cut} combinator introduced by Röjemo has the same
semantics as the \texttt{nofail} combinator introduced by Hutton.

Partridge and Wright \cite{partridge1996fv} showed that error detection
can be automated in parser combinators when we assume that the grammar
is LL(1).
Their main idea is:
if one alternative successfully consumes at least one symbol,
no other alternative can successfully consume any symbols.
Their technique is also known as the \emph{four-values} technique
because the parser finishes with one of the following values:
\texttt{Epsn}, when the parser finishes with success without
consuming any input;
\texttt{OK}, when the parser finishes with success consuming
some input;
\texttt{Fail}, when the parser fails without consuming any input; and
\texttt{Error}, when the parser fails consuming some input.
Three values were inspired by Hutton's work \cite{hutton1992hfp},
but with new meanings.

In the four-values technique, we do not need to annotate the grammar
because the authors changed the semantics of the sequence and choice
combinators to automatically generate the \texttt{Error} value
according to the Table \ref{tab:fv}.
In summary, the sequence combinator propagates an error when the
second parse fails after consuming some input while
the choice combinator does not try further alternatives
if the current one consumed at least one symbol from the input.
In case of error, the four-values technique detects the first symbol
following the longest parse of the input and uses this symbol to
report an error.

\begin{table}[t]
\begin{center}
\begin{tabular}{p{1cm}p{2cm}p{2cm}p{2cm}p{2cm}}
\hline
& $p_1$ & $p_2$ & $p_1 p_2$ & $p_1 \; | \; p_2$\\
\hline
& \texttt{Error} & \texttt{Error} & \texttt{Error} & \texttt{Error} \\
& \texttt{Error} & \texttt{Fail} & \texttt{Error} & \texttt{Error} \\
& \texttt{Error} & \texttt{Epsn} & \texttt{Error} & \texttt{Error} \\
& \texttt{Error} & \texttt{OK} ($x$) & \texttt{Error} & \texttt{Error} \\
& \texttt{Fail} & \texttt{Error} & \texttt{Fail} & \texttt{Error} \\
& \texttt{Fail} & \texttt{Fail} & \texttt{Fail} & \texttt{Fail} \\
& \texttt{Fail} & \texttt{Epsn} & \texttt{Fail} & \texttt{Epsn} \\
& \texttt{Fail} & \texttt{OK} ($x$) & \texttt{Fail} & \texttt{OK} ($x$) \\
& \texttt{Epsn} & \texttt{Error} & \texttt{Error} & \texttt{Error} \\
& \texttt{Epsn} & \texttt{Fail} & \texttt{Fail} & \texttt{Epsn} \\
& \texttt{Epsn} & \texttt{Epsn} & \texttt{Epsn} & \texttt{Epsn} \\
& \texttt{Epsn} & \texttt{OK} ($x$) & \texttt{OK} ($x$) & \texttt{OK} ($x$) \\
& \texttt{OK} ($x$) & \texttt{Error} & \texttt{Error} & \texttt{OK} ($x$) \\
& \texttt{OK} ($x$) & \texttt{Fail} & \texttt{Error} & \texttt{OK} ($x$) \\
& \texttt{OK} ($x$) & \texttt{Epsn} & \texttt{OK} ($x$) & \texttt{OK} ($x$) \\
& \texttt{OK} ($x$) & \texttt{OK} ($y$) & \texttt{OK} ($y$) & \texttt{OK} ($x$) \\
\end{tabular}
\end{center}
\caption{Behavior of sequence and choice in the four-values technique}
\label{tab:fv}
\end{table}

The four-values technique assumes that the input is composed by tokens
which are provided by a separate lexer.
However, being restricted to LL(1) grammars can be a limitation
because parser combinators, like PEGs, usually operate on strings of
characters to implement both lexer and parser together.
For instance, a parser for Tiny that is implemented with Parsec
\cite{leijen2001parsec} does not parse the following program:
\verb'read x;'.
That is, the matching of \verb'read' against \verb'repeat' generates
an error.
Such behavior is confirmed in Table \ref{tab:fv} by the third line
from the bottom.

Parsec is a parser combinator library for Haskell that employs a
technique equivalent to the four-values technique for implementing
LL(1) predictive parsers that automatically report errors
\cite{leijen2001parsec}. To overcome the LL(1) limitation,
Parsec introduced the {\tt try} combinator, a dual of Hutton's
{\tt nofail} combinator. The effect of {\tt try}
is to translate an error into a backtrackeable
failure. The idea is
to use {\tt try} to annotate the parts of the grammar where
arbitrary lookahead is needed. 

Parsec's restriction to LL(1) grammars made it possible to implement
an error reporting technique similar to the one
used in top-down parsers.
Parsec produces error messages that include the error position,
the character at this position and the \texttt{FIRST} and
\texttt{FOLLOW} sets of the
productions that were expected at this position.
Parsec also implements the error injection combinator (\texttt{<?>})
for naming productions.
This combinator gets two arguments: a parser \texttt{p} and a
string \texttt{exp}.
The string \texttt{exp} replaces the \texttt{FIRST} set of a
parser \texttt{p} when all the alternatives of \texttt{p} failed.
This combinator is useful to name terminals and non-terminals to
get better information about the context of a syntax error.

Swierstra and Duponcheel~\cite{swierstra1996dec} showed an
implementation of parser combinators for error recovery,
although most libraries and parser generators that are based
on parser combinators implement only error reporting.
Their work relies on the fact that the grammar is LL(1)
and shows an implementation of parser combinators that
repairs erroneous inputs, produces an appropriated message,
and continues parsing the rest of the input.
This approach was later extended to also deal with grammars
that are not LL(1), including ambiguous grammars~\cite{swierstra2009uuparsing}.
The extended approach relies heavily on some features that the
implementation language should have, such as lazy evaluation.

\section{Labeled Failures} \label{sec:lf}

Exceptions are a common mechanism for signaling and handling
errors in programming languages.
Exceptions let programmers classify the different errors their
programs may signal by using distinct types for distinct errors,
and decouple error handling from regular program logic.

In this section we add \emph{labeled failures} to PEGs,
a mechanism akin to exceptions and exception handling,
with the goal of improving error reporting while preserving
the composability of PEGs. In the next section we discuss how to use PEGs with labeled
failures to implement some of
the techniques that we have discussed in Section~\ref{sec:rel}: 
the \texttt{nofail} combinator~\cite{hutton1992hfp},
the \texttt{cut} combinator~\cite{rojemo1995epc},
the four-values technique~\cite{partridge1996fv} and
the \texttt{try} combinator~\cite{leijen2001parsec}.

A labeled PEG $G$ is a tuple $(V,T,P,L,p_{S})$ where
$L$ is a finite set of labels that must include the {\tt fail}
label, and the
expressions in $P$ have been extended with the {\em throw}
operator, explained below. The other parts use the same definitions from Section \ref{sec:pegs}. 

The abstract syntax of labeled parsing expressions adds the
\emph{throw} operator $\throw^{l}$, which generates a failure
with label $l$, and adds an extra argument $S$ to the ordered choice
operator, which is the set of labels that the ordered choice should
catch. $S$ must be a subset of $L$.
\[\it{
p = \varepsilon \; | \; a \; | \; A \; | \; p_1 p_2 \; | \;
  p_1 /^{S} p_2 \; | \; p* \; | \; !p \; | \; \throw^{l}
}\]

\begin{figure}
{\small
\begin{align*}
& \textbf{Empty} \;\;\;
{\frac{}{G[\varepsilon] \; x \Lp x}} \mylabel{empty.1}
\\ \\
& \textbf{Terminal} \;\;\;
{\frac{}{G[a] \; ax \Lp x}} \mylabel{char.1}
\;\;\;
{\frac{}{G[b] \; ax \Lp {\tt fail}}}, b \ne a \mylabel{char.2}
\;\;\;
{\frac{}{G[a] \; \varepsilon \Lp {\tt fail}}} \mylabel{char.3}
\\ \\
& \textbf{Non-terminal} \;\;\;
{\frac{G[P(A)] \; x \Lp X}
      {G[A] \; x \Lp X}} \mylabel{var.1}
\\ \\
& \textbf{Concatenation} \;\;\;
{\frac{G[p_1] \; xy \Lp y \fivespaces G[p_2] \; y \Lp X}
      {G[p_1 \; p_2] \; xy \Lp X}} \mylabel{con.1}
\;\;\;
{\frac{G[p_1] \; x \Lp l}
      {G[p_1 \; p_2] \; x \Lp l}} \mylabel{con.2}
\\ \\
& \textbf{Ordered Choice} \;\;\;
{\frac{G[p_1] \; xy \Lp y}
      {G[p_1 \;\slash^S\; p_2] \; xy \Lp y}} \mylabel{ord.1}
\;\;\;
{\frac{G[p_1] \; x \Lp l}
      {G[p_1 \;\slash^S\; p_2] \; x \Lp l}}, l \not\in S \mylabel{ord.2}
\\ \\
&
\tenspaces \tenspaces \tenspaces \,
{\frac{G[p_1] \; x \Lp l \fivespaces G[p_2] \; x \Lp X}
      {G[p_1 \;\slash^S\; p_2] \; x \Lp X}}, l \in S \mylabel{ord.3}
\\ \\
& \textbf{Repetition} \;\;\;
{\frac{G[p] \; x \Lp {\tt fail}}
      {G[p*] \; x \Lp x}} \mylabel{rep.1}
\;\;\;
{\frac{G[p] \; xy \Lp y \fivespaces G[p*] \; y \Lp X}
      {G[p*] \; xy \Lp X}} \mylabel{rep.2}
\\ \\
&
\tenspaces \tenspaces \;\,
{\frac{G[p] \; x \Lp l}
      {G[p*] \; x \Lp l}}, l \neq {\tt fail} \mylabel{rep.3}
\\ \\
& \textbf{Negative Predicate} \;\;\;
{\frac{G[p] \; x \Lp {\tt fail}}
      {G[!p] \; x \Lp x}} \mylabel{not.1}
\;\;\;
{\frac{G[p] \; xy \Lp y}
      {G[!p] \; xy \Lp {\tt fail}}} \mylabel{not.2}
\\ \\
&
\tenspaces \tenspaces \tenspaces \fivespaces \;
{\frac{G[p] \; x \Lp l}
      {G[!p] \; x \Lp l}}, l \neq {\tt fail} \mylabel{not.3}
\\ \\
& \textbf{Throw} \;\;\;
{\frac{}{G[\throw^{l}] \; x \Lp l}} \mylabel{throw.1}
\end{align*}
}
\caption{Semantics of PEGs with labels}
\label{fig:sem}
\end{figure}

Figure \ref{fig:sem} presents the semantics of PEGs with labels
as a set of inference rules.
The semantics of PEGs with labels is defined by the relation
$\Lp$ among a parsing expression, an input string and a result.
The result is either a string or a label.
The notation $G[p] \; xy \Lp \; y$ means that the expression
$p$ matches the input $xy$, consumes the prefix $x$ and leaves the
suffix $y$ as the output.
The notation $G[p] \; xy \Lp \; l$ indicates that the
matching of $p$ fails with label $l$ on the input $xy$.

Now $a$ matches and consumes itself and fails with label \texttt{fail}
otherwise;
$p_1 p_2$ tries to match $p_1$, if $p_1$ matches an input prefix,
then it tries to match $p_2$ with the suffix left by $p_1$,
the label $l$ is propagated otherwise;
$p_1 /^{S} p_2$ tries to match $p_1$ in the input and tries to
match $p_2$ in the same input only if $p_1$ fails with a label
$l \in S$, the label $l$ is propagated otherwise;
$p*$ repeatedly matches $p$ until the matching of $p$
silently fails with label {\tt fail}, and propagates a label $l$
when $p$ fails with this label;
$!p$ successfully matches if the input does not match $p$ with
the label {\tt fail}, fails producing the label \texttt{fail} when
the input matches $p$, and propagates a label $l$ when $p$
fails with this label, not consuming the input in all cases;
$\throw^{l}$ produces the label $l$.

We faced some design decisions in our formulation that are worth
discussing. First, we use a set of labels in the ordered choice as
a convenience. We could have each ordered choice handling a single label,
and it would just lead to duplication:
an expression
$p_1 \; /^{\{l_1,l_2,...,l_n\}} \; p_2$
would become
$( \; ... \; ((p_1 \; /^{l_1} \; p_2) \; /^{l_2} \; p_2) \; ... \; /^{l_n} \; p_2)$.

Second, we require the presence of a \texttt{fail} label
to maintain compatibility with the original semantics of PEGs,
where we only have \texttt{fail} to signal both error and failure.
For the same reason, we define the expression $p_1 / p_2$ as
syntactic sugar for $p_1 /^{\{{\tt fail}\}} p_2$.

Another choice was how to handle labels in a repetition.
We chose to have a repetition stop silently only on the
\texttt{fail}
label to maintain the following identity, which holds
for unlabeled PEGs: an expression $p*$
is equivalent to a fresh non-terminal $A$ with the rule
$A \leftarrow p \; A \; / \; \varepsilon$.

Finally, the negative predicate succeeds only on the
 \texttt{fail} label to
allow the implementation of the positive predicate: the
 expression
$\&p$ that implements the positive predicate in the original semantics
of PEGs \cite{ford2002packrat,ford2002pappy,ford2004peg}
is equivalent to the expression $!!p$.
Both expressions successfully match if the input matches $p$,
fail producing the label \texttt{fail} when the input does
not match $p$, and propagate a label $l$ when $p$ fails
with this label, not consuming the input in all cases.

Figure \ref{fig:tinylabels} presents a PEG with labels for the Tiny
language from Section \ref{sec:pegs}. 
The expression $[p]^{l}$ is syntactic
sugar for $(p \; / \; \throw^{l})$. 

The strategy we used to annotate the grammar was the following:
first, annotate every terminal that should not fail, that is,
making the PEG backtrack on failure of that terminal would
be useless, as the whole parse would either fail or not consume
the whole input in that case. For an LL(1) grammar like the
one in the example, that means all terminals in a production except
the one in the very beginning of the production.

After annotating the terminals, we do the same for whole productions.
We annotate productions where failing the whole production 
always implies an error in the input, adding a new alternative that
throws an error label specific to that production.

For Tiny,
we end up annotating just two productions, {\em Factor} and {\em CmdSeq}.
Productions {\em Exp}, {\em SimpleExp}, and {\em Term} also
should not fail, but after annotating {\em Factor} they always
either succeed or throw the label {\tt exp}. The {\it Cmd}
production can fail, because it controls whether the repetition inside
{\it CmdSeq} stops or continue.

Notice that this is just an example of how a grammar can be
annotated. More thorough analyses are possible: for example,
we can deduce that {\it Cmd} is not allowed to fail unless
the next token is one of {\tt ELSE}, {\tt END}, {\tt UNTIL},
or the end of the input (the {\tt FOLLOW} set of {\it Cmd}), and instead of $\throw^{cmd}$ add $!({\tt ELSE} \; / \; {\tt END} \; / \;
{\tt UNTIL} \; / \; !.) \throw^{cmd}$ as a new alternative. This would
remove the need for the $\throw^{cmd}$ annotation of {\it CmdSeq}.

\begin{figure}
\begin{align*}
\it{Tiny} & \leftarrow \it{CmdSeq}\\
\it{CmdSeq} & \leftarrow \it{(Cmd \; [{\tt SEMICOLON}]^{\tt sc}) \;
  (Cmd \; [{\tt SEMICOLON}]^{\tt sc})*} \; / \; \throw^{cmd}\\
\it{Cmd} & \leftarrow \it{IfCmd \; / \; RepeatCmd \; / \;
  AssignCmd \; / \; ReadCmd \; / \; WriteCmd}\\
\it{IfCmd} & \leftarrow \it{{\tt IF} \, Exp \,
  [{\tt THEN}]^{\tt then} \, CmdSeq \,
  ({\tt ELSE} \; CmdSeq \;/\; \varepsilon) \;
  [{\tt END}]^{\tt end}}\\
\it{RepeatCmd} & \leftarrow \it{{\tt REPEAT} \; CmdSeq \;
  [{\tt UNTIL}]^{\tt until} \; Exp}\\
\it{AssignCmd} & \leftarrow \it{{\tt NAME} \; [{\tt ASSIGNMENT}]^{\tt bind} \;
  Exp}\\
\it{ReadCmd} & \leftarrow \it{{\tt READ} \; [{\tt NAME}]^{\tt read}}\\
\it{WriteCmd} & \leftarrow \it{{\tt WRITE} \; Exp}\\
\it{Exp} & \leftarrow \it{SimpleExp \;
  (({\tt LESS} \; / \; {\tt EQUAL}) \; SimpleExp \; /
  \; \varepsilon})\\
\it{SimpleExp} & \leftarrow \it{Term \;
  (({\tt ADD} \; / \; {\tt SUB}) \; Term)*}\\
\it{Term} & \leftarrow \it{Factor \;
  (({\tt MUL} \; / \; {\tt DIV}) \; Factor)*}\\
\it{Factor} & \leftarrow \it{{\tt OPENPAR} \; Exp
  \; [{\tt CLOSEPAR}]^{\tt cp} \; /
  \; {\tt NUMBER} \; / \; {\tt NAME}} \; / \; \throw^{\tt exp}
\end{align*}
\caption{A PEG with labels for the Tiny language}
\label{fig:tinylabels}
\end{figure}

The PEG reports an error when parsing finishes with an
uncaught label. Each label is associated with a
meaningful error message.
For instance, if we use this PEG for Tiny to parse the
code example from Section \ref{sec:pegs}, parsing finishes with the
\texttt{sc} label and the PEG can use it to produce the following
error message:

\begin{verbatim}
    factorial.tiny:6:1: syntax error, there is a missing ';'
\end{verbatim}

Note how the semantics of the repetition works with the rule
\textit{CmdSeq}.
Inside the repetition, the \texttt{fail} label means that there are no
more commands to be matched and the repetition should stop while the
\texttt{sc} label means that a semicolon (\texttt{;}) failed to match.
It would not be possible to write the rule \textit{CmdSeq}
using repetition if we had chosen to stop the repetition with any
label, instead of stopping only with the \texttt{fail} label,
because the repetition would accept the \texttt{sc} label as the
end of the repetition whereas it should propagate this label.

Although the semantics of PEGs with labels presented in
Figure~\ref{fig:sem} allows us to generate specific error
messages, it does not give us information about the location
where the failure probably is, so it is necessary to use some
extra mechanism (e.g., semantic actions) to get this information.
To avoid this, we can adapt the semantics of PEGs with labels
to give us a tuple $\Tup{l}{y}$ in case of a failure, where
$y$ the suffix of the input that PEG was trying to match
when label $l$ was thrown. Updating the semantics of Figure~\ref{fig:sem}
to reflect this change is straightforward.

In the next section, we try to establish a comparison between
the farthest failure position heuristic and the labeled failure mechanism
by contrasting two different implementations of a parser for a dialect of
the Lua language.

\section{Labeled Failures versus Farthest Failure Position}
\label{sec:labelsfft}

In this section we will compare two parser implementations for
the Typed Lua language, one that uses the farthest failure
position heuristic for error reporting, which was implemented
first, and one based on labeled failures.

Typed Lua~\cite{maidl2014typedlua} is an optionally-typed
extension of the Lua programming language~\cite{lua}.
The Typed Lua parser recognizes plain Lua programs, and
also Lua programs with type annotations. The first version
of the parser was implemented using Ford's
heuristic and the LPeg library~\footnote{The first version of Typed Lua parser is available at
\url{https://github.com/andremm/typedlua/blob/master/typedlua/tlparser.lua}}.

As LPeg does not have a native error reporting mechanism based
on Ford's strategy, the failure tracking heuristic was implemented
following the approach described in Section~\ref{sec:pegs},
which uses semantic actions.

Below we have the example of a Lua statement with a syntax error:
\begin{verbatim}
     a = function (a,b,) end
\end{verbatim}

In this case, the parser gives us the following error message,
which is quite precise:
\begin{verbatim}
    test.lua:1:19: syntax error, unexpected ')', expecting '...', 'Name'
\end{verbatim}

In the previous case, the list of expected tokens had only two
candidates, but this is not always the case. For example, let us
consider the following Lua program, where there is no expression
after the \textbf{elseif} in line 5:
\begin{verbatim}
 01  if a then
 02    return x
 03  elseif b then
 04    return y
 05  elseif
 06
 07  end
\end{verbatim}

The corresponding error message has a lengthy list of tokens,
which does not help much to fix the error:
\begin{verbatim}
    test.lua:7:1: syntax error, unexpected 'end', expecting '(', 'Name', '{', 
                  'function', '...', 'true', [9 more tokens] 
\end{verbatim}

When using the Typed Lua parser based on Ford's heuristic it is not
uncommon to get a message like this. An analysis of the test cases
available in the parser package shows us that around half of the
expected error messages have a list of at least eight expected tokens
(there were messages with a list of 39 expected tokens).

The second implementation of the Typed Lua parser was 
based on labeled failures and used the LPegLabel library~\cite{lpeglabel},
which is an extension of the LPeg library that supports
labeled failures~\footnote{The second version of Typed Lua parser is available at
\url{https://github.com/sqmedeiros/lpeglabel/tree/master/examples/typedlua}}.

The use of labeled failures implies in an annotation burden to
specify when each label should be thrown. In the case of Typed Lua grammar,
we defined almost 50 different labels, using the same basic strategy
that we used to annotate the Tiny grammar of Section~\ref{sec:lf}.

Given the previous Lua program, the error message presented now is:
\begin{verbatim}
    test.lua:7:1: expecting <exp> after 'elseif'
\end{verbatim}

This error message is more informative than the previous one,
which was generated automatically. We analyzed the error
messages generated by the two parsers in 53 examples,
and considered that in more than half of these examples
the parser based on labeled failures produced a better
error message. In about 20\% of the cases we considered
the error messages of both approaches similar, and in other
20\% of the cases the parser based on Ford's heuristic
generated better error messages.

The error location indicated by the two parsers in the examples
analyzed was essentially the same. This seems to indicate that
the main difference in practice between both approaches is related
to the length of the error message generated. By using labeled
failures we can probably get a simple error message at the cost
of annotating the grammar, while by using the farthest failure
tracking heuristic we can automatically generate error messages,
which sometimes may contain a long list of expected tokens.

A point that is worth mentioning about the labeled failure approach
is that it is not mandatory to annotate the entire grammar. The grammar
can be annotated incrementally, at the points where the current error
message is not good enough, and when no specific label is thrown, i.e.,
when the label \texttt{fail} is thrown, an error message can be
generated automatically by using the position where the failure
occurred. This means that combining labeled failures with the farthest
failure position reduces the annotation burden, and helps to
identify the places in the parser where the a label is desirable.

In the next
section, we discuss some applications of labeled failures:
we can use labeled PEGs to express the error
reporting techniques that we have discussed in
Section~\ref{sec:rel}~\cite{hutton1992hfp,rojemo1995epc,partridge1996fv,leijen2001parsec}, and also to efficiently parse
context-free grammars that can use the
$LL(*)$ parsing strategy~\cite{parr2011llstar}.

\section{Applications of Labeled Failures}
\label{sec:labelsrelated}

This section shows that PEGs with labeled failures can express
several error reporting techniques used in the realm of parsing
combinators. They can also efficiently parse
context-free grammars that
are parseable by the $LL(*)$ top-down parsing strategy.

In Hutton's deterministic parser combinators~\cite{hutton1992hfp}, the \texttt{nofail}
combinator is used to distinguish between failure and error.
We can express the \texttt{nofail} combinators using PEGs with labels
as follows:
\begin{align*}
{\tt nofail} \; p & \; \equiv \; p \; / \; \throw^{\tt error}
\end{align*}

That is, \texttt{nofail} is an expression that transforms the failure
of $p$ into an error to abort backtracking.
Note that the \texttt{error} label should not be caught by any
ordered choice.
Instead, the ordered choice propagates this label and catches solely
the \texttt{fail} label.
The idea is that parsing should finish with one of the following
values: \texttt{success}, \texttt{fail} or \texttt{error}.

The annotation of the Tiny grammar to use \texttt{nofail} is
similar to the annotation we have done using labeled failures.
Basically, we just need to change the grammar to use
\texttt{nofail} instead of $[p]^{l}$.
For instance, we can write the rule \textit{CmdSeq} as follows:
\[
\it{CmdSeq} \leftarrow \it{(Cmd \; ({\tt nofail \; SEMICOLON})) \;
  (Cmd \; ({\tt nofail \; SEMICOLON}))*}\\
\]

If we are writing a grammar from scratch, there is no advantage to
use \texttt{nofail} instead of more specific labels,
as the annotation burden is the same and with \texttt{nofail} we
lose more specific error messages.

The \texttt{cut} combinator~\cite{rojemo1995epc} was introduced
to reduce the space inefficiency of backtracking parsers, where
the possibility of backtracking implies that any input that has already
been processed must be kept in memory until the end of parsing.
Semantically it is identical to {\tt nofail}, differing only
in the way the combinators are implemented: to implement {\tt cut}
the parser combinators use continuation-passing style, so {\tt cut}
can drop the failure continuation and consequently any pending
backtrack frames. Hutton's {\tt nofail} is implemented in direct style,
and is not able to drop pending backtrack frames.
Expressing a {\tt cut} operator with the same properties
is not possible in our semantics of PEGs.

The four-values technique changed the semantics
of parser combinators to implement predictive parsers for LL(1)
grammars that automatically identify the longest input prefix
in case of error, without needing annotations in the grammar.
We can express this technique using labeled failures by transforming
the original PEG with the following rules:
\begin{align}
\label{fv:e}
\llbracket \varepsilon \rrbracket & \equiv \; \throw^{\tt epsn}\\
\label{fv:a}
\llbracket a \rrbracket & \equiv \; a\\
\label{fv:A}
\llbracket A \rrbracket & \equiv \; A\\
\label{fv:bind}
\llbracket p_1 p_2 \rrbracket & \equiv \;
\llbracket p_1 \rrbracket \;
(\llbracket p_2 \rrbracket \; / \throw^{\tt error} \; /^{\{{\tt epsn}\}} \;
\varepsilon) \; /^{\{{\tt epsn}\}} \; \llbracket p_2 \rrbracket\\
\label{fv:choice}
\llbracket p_1 / p_2 \rrbracket & \equiv \;
\llbracket p_1 \rrbracket \; /^{\{{\tt epsn}\}} \;
(\llbracket p_2 \rrbracket \; / \throw^{\tt epsn}) \; / \;
\llbracket p_2 \rrbracket
\end{align}

This translation is based on three labels:
\texttt{epsn} means that the expression successfully finished
without consuming any input,
\texttt{fail} means that the expression failed without consuming
any input,
and \texttt{error} means that the expression failed after
consuming some input.
In our translation we do not have an \texttt{ok} label because a
resulting suffix means that the expression successfully finished
after consuming some input.
It is straightforward to check that the translated expressions
behave according to the Table \ref{tab:fv} from Section \ref{sec:rel}.

Parsec introduced the \texttt{try} combinator to annotate parts of the
grammar where arbitrary lookahead is needed.
We need arbitrary lookahead because PEGs and parser combinators usually
operate at the character level.
The authors of Parsec also showed a correspondence between the semantics
of Parsec as implemented in their library and Partridge and Wright's
four-valued combinators,
so we can emulate the behavior of Parsec using labeled failures by
building on the five rules above and adding the following rule for
\texttt{try}:
\begin{align}
\label{fv:try}
\llbracket {\tt try} \; p \rrbracket & \equiv
\llbracket p \rrbracket \; /^{\{{\tt error}\}} \; \throw^{\tt fail}
\end{align}

If we take the Tiny grammar of Figure \ref{fig:tiny} from
Section \ref{sec:pegs}, insert \texttt{try} in the necessary places,
and pass this new grammar through the transformation
$\llbracket\;\rrbracket$, then we get a PEG that automatically
identifies errors in the input with the \texttt{error} label.
For instance, we can write the rule \textit{RepeatCmd} as follows:
\begin{align*}
\it{RepeatCmd} & \leftarrow \it{({\tt try \; REPEAT}) \; CmdSeq \;
  {\tt UNTIL} \; Exp}
\end{align*}

$LL(*)$~\cite{parr2011llstar} is a parsing strategy used by
the popular parsing tool ANTLR~\cite{parr2013antlr,antlrsite}~\footnote{The
recently released version $4$ of ANTLR uses {\em adaptive}
 $LL(*)$
as its parsing strategy.}. An $LL(*)$ parser is a top-down
parser with arbitrary lookahead. The main idea of $LL(*)$
parsing is to build a deterministic finite automata for
each rule in the grammar, and use this automata to predict
which alternative of the rule the parser should follow,
based on the rest of the input. Each final state of
the DFA should correspond to a single alternative, or
we have an $LL(*)$ parsing conflict. 

Mascarenhas et al.~\cite{mascarenhas2014} shows how
CFG classes that correspond to top-down predictive parsing
strategies can be encoded with PEGs by using predicates to
encode the lookahead necessary for each alternative.
As translating a Deterministic Finite Automata (DFA) to a
PEG is straightforward~\cite{ierusalimschy2009lpeg,mascarenhas2014},
this gives us one way of encoding an $LL(*)$ parsing strategy in
a PEG, at the cost of encoding a different copy of the
lookahead DFA for each alternative.

Labeled PEGs provide a more straightforward encoding,
where instead of a predicate for each alternative, 
we use a single encoding of the lookahead DFA, where each
final state ends with a label corresponding to one
of the alternatives. Each alternative is preceded by
a choice operator that catches its label.

\begin{figure}[t]
\begin{center}
\begin{tikzpicture}[shorten >=1pt,node distance=2.9cm,on grid,auto] 
   \node[state,initial] (s_0)   {$s_0$}; 
   \node[state] (s_1) [above right=of s_0] {$s_1$}; 
   \node[state,accepting] (s_4) [above right=of s_1] {$s_4 \rightarrow 2$}; 
   \node[state,accepting] (s_6) [ right=of s_1] {$s_6 \to 1$}; 
   \node[state] (s_2) [ right=of s_0] {$s_2$}; 
   \node[state,accepting](s_5) [ below=of s_6] {$s_5 \rightarrow 4$};
   \node[state,accepting](s_3) [below=of s_5] {$s_3 \rightarrow 3$};
    \path[->] 
    (s_0) edge  node {ID} (s_1)
          edge[below]  node {unsigned} (s_2)
          edge[below]  node {int} (s_3)
    (s_1) edge  node  {=} (s_4)
          edge  node {EOF} (s_6)
          edge node  {ID}  (s_5)
    (s_2) edge  node {ID} (s_5)
          edge  node {int} (s_3) 
          edge [loop below] node {unsigned} ();
\end{tikzpicture}
\end{center}
\caption{$LL(*)$ lookahead DFA for rule $S$}
\label{fig:antlrdfa}
\end{figure}

To make the translation clearer, let us consider the following
example, from Parr and Fisher~\cite{parr2011llstar},
where non-terminal $S$ uses non-terminal $Exp$ (omitted) to match
arithmetic expressions:
\[
\it{S} \;\rightarrow\; {\tt ID} \;\,|\,\; \tt{ID} \,`{\tt =}\textrm'\, {\it Exp} \;\,|\;\,
                   `{\tt unsigned}\textrm' \; `\tt{*}\textrm' \, `{\tt int}\textrm'\; {\tt ID} \;\,|\;\,
                   `{\tt unsigned}\textrm' \;`\tt{*}\textrm'\, {\tt ID}\; {\tt ID}
\]

After analyzing this grammar, ANTLR produces the DFA
of Figure~\ref{fig:antlrdfa}. When trying to match $S$,
ANTLR runs this DFA on the input until it reaches a
final state that indicates which alternative of the
choice of rule $S$ should be tried. For example,
ANTLR chooses the second alternative if the DFA
reaches state $s_4$.

Figure~\ref{fig:pegllstar} gives a labeled PEG that encodes
the $LL(*)$ parsing strategy for rule $S$. Rules $S_0$,
$S_1$, and $S_2$ encode the lookahead DFA of Figure~\ref{fig:antlrdfa}, and correspond to states
$s_0$, $s_1$, and $s_2$, respectively. The throw
expressions correspond to the final states.
As the throw expressions make the input backtrack to where
it was prior to parsing $S_0$, we do not need to use a
predicate. We can also turn any uncaught failures into
errors.

\begin{figure}[t]
\begin{align*}
\it{S} & \leftarrow \it{S_0} \;/^{1}\; \tt{ID} \;/^{2}\;
	             \tt{ID} \,`\tt{=}\textrm'\, \it{Exp} \;/^{3}\; 
               `\tt{unsigned}\textrm' \; `\tt{*}\textrm' \; `\tt{int}\textrm' \; \tt{ID} \;/^{4}\;
               `\tt{unsigned}\textrm' \; `\tt{*}\textrm' \; \tt{ID}  \; \tt{ID}
\;/\; \throw^{\tt error} \\
S_0 & \leftarrow \tt{ID} \, \it{S_1}  \;/\; `\tt{unsigned}\textrm' \, \it{S_2}  \;/\;  `\tt{int}\textrm' \throw^3 \\
S_1 & \leftarrow \,`\tt{=}\textrm' \throw^2  \;/\;  !. \throw^1  \;/\;  \tt{ID} \throw^4  \\
S_2 & \leftarrow `\tt{unsigned}\textrm' \, \it{S_2}  \;/\;  \tt{ID} \throw^4  \;/\;  `\tt{int}\textrm' \throw^3
\end{align*}
\caption{PEG with Labels that Simulates the $LL(*)$ Algorithm}
\label{fig:pegllstar}
\end{figure}

\section{Conclusions} 
\label{sec:conc}

In this paper, we discussed error reporting strategies for
Parsing Expression Grammars. PEGs behave badly on the presence
of syntax errors, because backtracking often makes the PEG
lose track of the position where the error happened.
This limitation was already known by Ford, and he
tried to fix it in his PEG implementation by having
the implementation track the farthest position in the
input where a failure has happened~\cite{ford2002packrat}.

We took Ford's failure tracking heuristic and showed that
it is not necessary to modify a PEG implementation to
track failure positions as long as the implementation
has mechanisms to execute semantic actions, and the
current parsing position is exposed to these actions.
In addition, we also showed how it is easy to extend
the semantics of PEGs to incorporate failure tracking,
including information that can indicate what the
PEG was expecting when the failure happened.

Tracking the farthest failure position, either by
changing the PEG implementation, using semantic actions,
or redefining the semantics of PEGs, helps PEG parsers
produce error messages that are close to error messages
that predictive top-down parsers are able to produce,
but these are generic error messages, sometimes with
a long list of expected tokens.

As a way of generating more specific error messages,
we introduced a mechanism of labeled failures to PEGs.
This mechanism closely resembles standard exception
handling in programming languages. Instead of a
single kind of failure, we introduced a \emph{throw} operator
$\throw^{l}$ that can throw different kinds of failures,
identified by their labels, and extended the ordered choice
operator to specify the set of labels that it catches.
The implementation of these extensions in parser generator
tools based on PEGs is straightforward.

We showed how labeled failures can be used as a way to
annotate error points in a grammar, and tie them to more
meaningful error messages. Labeled failures are orthogonal
to the failure tracking approach we discussed earlier, so
grammars can be annotated incrementally, at the points
where better error messages are judged necessary.

We also showed that the labeled failures approach can express
several techniques for error reporting used in parsers
based on deterministic parser combinators, as
presented in related
work~\cite{hutton1992hfp,rojemo1995epc,partridge1996fv,leijen2001parsec}.
Labeled failures can also be used as a way of encoding 
the decisions made by a predictive top-down parser, as long
as the decision procedure can be encoded as a PEG, and showed
an example of how to encode an $LL(*)$ grammar in this way.

Annotating a grammar with labeled failures demands care:
if we mistakenly annotate expressions that should be able to
fail, this modifies the behavior of the parser beyond error 
reporting. In any case, the use of labeled PEGs for error
reporting introduces an annotation burden that is
lesser than the annotation burden introduced by error
productions in LR parsers, which also demand care, as
their introduction usually lead to \textit{reduce-reduce}
conflicts~\cite{jeffery2003lre}.

We showed the error reporting strategies in the context
of a small grammar for a toy language, and we also discussed
the implementation of parsers for the Typed Lua language,
an extension of the Lua language, based on these strategies.
Moreover, we also implemented parsers for other languages,
such as Céu~\cite{ceu}, based on theses approaches,
improving the quality of error reporting either with
generic error messages or with more specific error messages.

\bibliographystyle{elsarticle-num}
\bibliography{peg-labels-scp}

\appendix
\appendix
\section{Formalization of Farthest Failure Tracking}
\label{sec:semfarther}

In Section~\ref{sec:pegs}, we saw that a common approach
to report errors in PEGs is to inform the position of the
farthest failure,
and how we can use semantic actions to implement
this approach in common PEG implementations. In this appendix,
we show a conservative extension of the PEG formalism
that expresses this error reporting approach directly,
making it a part of the semantics of PEGs instead of an
ad-hoc extension to PEG implementations.

The result of a PEG on a given input is either
$\Fail$ or a suffix of the input. Returning a suffix of the
input means that the PEG succeeded and consumed a prefix
of the input. PEGs with farthest failure tracking return
the product of the original result and either another
suffix of the input, to denote the position of the farthest
failure during parsing, or $\Nil$ to denote that 
there were no failures. 

Figure~\ref{fig:semfarthest} presents a formal semantics of
PEGs with farthest failure tracking as a set of inference rules.
The notation $G[p] \; xy \Lp \; \Tup{y}{v?}$ represents
a successful match of the parsing expression $p$ in the context of
a PEG $G$ against the subject $xy$, consuming $x$ and leaving the
suffix $y$.
In a similar way, $G[p] \; xy \Lp \; \Tup{\Fail}{v?}$
represents an unsuccessful match. The term $v?$ denotes
a value that is either a suffix $v$ of the {\em original} input or
$\Nil$. We cannot use an empty string to denote that there were no failures,
as an empty string already means that a failure has occurred after
the last position of the input.
The term $\Any$ denotes a value that is either a suffix of the
{\em current} input or $\Fail$.

The auxiliary function \Suf\ that appears on
Figure~\ref{fig:semfarthest} compares two possible
error positions, denoted by a suffix of the
input string, or \Nil\, if no failure has occurred,
and returns the furthest: any suffix of the input
is a further possible error position than \Nil\,
and a shorter suffix is a further possible
error position than a longer suffix.

Rule {\bf empty.1} deals with the empty expression. This
expression never fails, so the failure position is always
$\Nil$. Rule {\bf var.1} deals with
non-terminals, so it just propagates the result of
matching the right-hand side of the non-terminal.
Rules {\bf char.1}, {\bf char.2}, and {\bf char.3} deal with
terminals. The latter two rules denote failures, so they
return the subject as the failure position.

\begin{figure}[t]
	{\small
		\begin{align*}
		& \textbf{Empty} \;\;\;
		{\frac{}{G[\varepsilon] \; x \Lp \Tup{x}{\Nil}}} \mylabel{empty.1}
		\fivespaces \textbf{Non-terminal} \;\;\;
		{\frac{G[P(A)] \; x \Lp \Tup{\Any}{\Ex}}
			{G[A] \; x \Lp \Tup{\Any}{\Ex}}} \mylabel{var.1}
		\\ \\
		& \textbf{Terminal} \tenspaces 
		{\frac{}{G[a] \; ax \Lp \Tup{x}{\Nil}}} \mylabel{char.1} \\ \\
		& \twentyspaces
		{\frac{}{G[b] \; ax \Lp \Tup{\Fail}{ax}}}, b \ne a \mylabel{char.2}
		\fivespaces
		{\frac{}{G[a] \; \varepsilon \Lp \Tup{\Fail}{\varepsilon}}} \mylabel{char.3}
		\\ \\
		& \textbf{Concatenation} \;\;\;
		{\frac{G[p_1] \; xy \Lp \Tup{y}{\Ex} \fivespaces G[p_2] \; y \Lp \Tup{\Any}{\Ey}}
			{G[p_1 \; p_2] \; xy \Lp \Tup{\Any}{\Suff{\Ex}{\Ey}}}}   \mylabel{con.1}
		\;\;\;
		{\frac{G[p_1] \; x \Lp \Tup{\Fail}{\Ex}}
			{G[p_1 \; p_2] \; x \Lp \Tup{\Fail}{\Ex}}} \mylabel{con.2}
		\\ \\
		& \textbf{Ordered Choice} \;\;\;
		{\frac{G[p_1] \; xy \Lp \Tup{y}{\Ex}}
			{G[p_1 \;/\; p_2] \; xy \Lp \Tup{y}{\Ex}}} \mylabel{ord.1}
		\;\;\;
		{\frac{G[p_1] \; x \Lp \Tup{\Fail}{\Ex} \fivespaces G[p_2] \; x \Lp \Tup{\Any}{\Ey}}
			{G[p_1 \;/\; p_2] \; x \Lp \Tup{\Any}{\Suff{\Ex}{\Ey}}}}  \mylabel{ord.2}
		\\ \\
		& \textbf{Repetition} \;\;\;
		{\frac{G[p] \; x \Lp \Tup{\Fail}{\Ex}}
			{G[p*] \; x \Lp \Tup{x}{\Ex}}} \mylabel{rep.1}
		\;\;\;
		{\frac{G[p] \; xyz \Lp \Tup{yz}{\Ex} \fivespaces G[p*] \; yz \Lp \Tup{z}{\Ey}}
			{G[p*] \; xyz \Lp \Tup{z}{\Suff{\Ex}{\Ey}}}} \mylabel{rep.2}
		\\ \\
		& \textbf{Negative Predicate} \;\;\;
		{\frac{G[p] \; x \Lp \Tup{\Fail}{\Ex}}
			{G[!p] \; x \Lp \Tup{x}{\Nil}}} \mylabel{not.1}
		\;\;\;
		{\frac{G[p] \; xy \Lp \Tup{y}{\Ex}}
			{G[!p] \; xy \Lp \Tup{\Fail}{xy}}} \mylabel{not.2}
		\end{align*}
	}
	\caption{Semantics of PEGs with farthest failure tracking}
	\label{fig:semfarthest}
\end{figure}

Rules {\bf con.1} and {\bf con.2} deal with concatenation.
The second rule just propagates the failure position, but
rule {\bf con.1} needs to take the farthest position
between the two parts of the concatenation. The rules
for ordered choice ({\bf ord.1} and {\bf ord.2}) and
repetition ({\bf rep.1} and {\bf rep.2}) work in
a similar way: whenever there are two possible farthest
failure positions, we use \Suf\ to take the farthest of
them.

Finally, rules {\bf not.1} and {\bf not.2} deal with
the syntactic predicate. Fitting this predicate in
this error reporting framework is subtle. The rules
not only need to make sense for the predicate in
isolation, but also have to make sense for $!!p$, which
works as a ``positive lookahead" predicate that succeeds
if $p$ succeeds but fails if $p$ fails while never
consuming any input.

For {\bf not.1}, the case where $!p$ succeeds, we have two
choices: either propagate $p$'s farthest failure position, or
ignore it (using $\Nil$). The first choice can lead to an
odd case, where the failure that made $!p$ succeed can get
the blame for the overall failure of the PEG, so
{\bf not.1} takes the second choice.

We also have two choices for the
case where $!p$ fails: either propagate the failure position
from $p$ or just use the current position. The first choice
can also lead to an odd cases where the overall failure
of the PEG might be blamed on something in $p$ that, if
corrected, still makes $p$ succeed and $!p$ fail, so
{\bf not.2} also takes the second choice.

The end result is that what happens inside a predicate simply
does not take part in error reporting at all, which is the
simplest approach, and also gives a consistent result for
$!!p$.

As an example that shows the interaction of these
rules, let us consider again the Tiny program in
Figure~\ref{fig:tiny}, reproduced below:

\begin{verbatim}
01  n := 5;
02  f := 1;
03  repeat
04    f := f * n;
05    n := n - 1
06  until (n < 1);
07  write f;
\end{verbatim}

The missing ``{\tt ;}" at the end of line 5
makes the repetition ({\it Cmd {\tt SEMICOLON}})*
inside {\it CmdSeq} succeed through rule {\bf rep.1}
and propagate the failure position to the concatenation
inside this same non-terminal. Rule {\bf con.2} then
propagates the same failure position to the concatenation
inside {\it RepeatCmd}.

Another failure occurs in {\it RepeatCmd} when {\tt until}
does not match {\tt n}, but rule {\bf con.2} again propagates
the position of the missing ``{\tt ;}", which is farthest.
This leads to a failure inside the repetition on {\it CmdSeq}
again, which propagates the position of the missing ``{\tt ;}"
through rule {\bf rep.2}. Finally, rules {\bf con.2} and {\bf var.1}
propagate this failure position to the top of the grammar
so the missing semicolon gets the blame for the PEG
not consuming the whole program.

We can translate the failure suffix into a line and column
number inside the original subject, and extract the
first token from the beginning of this suffix to
produce an error message similar to the error messages
in Section~\ref{sec:pegs}:

\begin{verbatim}
factorial.tiny:6:1: syntax error, unexpected 'until'
\end{verbatim}

While this message correctly pinpoints the location
of the error, it could be more informative.
\ref{sec:fartherandlist} shows how we can
extend the semantics of Figure~\ref{fig:semfarthest}
to gather more information than just the farthest
failure position, thus making us able to generate
error messages as informative as the ones produced by Pappy,
which we discussed in the end of Section~\ref{sec:pegs}.
In the next section, we will introduce another approach
for error reporting in PEGs, which can produce more precise
error messages, at the cost of annotating the grammar.

\section{Formalization of Farthest Failure Tracking and Error Lists}
\label{sec:fartherandlist}

In this appendix we show how the formalization of the farthest failure tracking
from \ref{sec:semfarther} can easily be adapted
to build a more elaborate error object, instead of just returning a position.

The new semantics formalizes a strategy similar to
the one used by Ford~\citep{ford2002packrat} in his
PEG implementation. The basic idea is to keep a list of
the simple expressions that the PEG was expecting to match when
a failure occurred, and use this list to build
an error message. Figure~\ref{fig:semfarthestjoin}
gives inference rules for the $\Lp$ relation of
this new semantics.

The result of the matching of a PEG $G$ against an input $x$
is still a pair, but the second component is now
another pair, the {\em error pair}.
The first component of the error pair is the farthest
error position, same as in the previous semantics.
The second component is a list of parsing expressions
that were expected at this error position.
If the grammar does not have syntactic predicates,
the expressions in this list are just terminals and
non-terminals.

\begin{figure}[p]
{\small
\begin{align*}
& \textbf{Empty} \;\;\;
{\frac{}{G[\varepsilon] \; x \Lp \Tup{x}{(\Nil,\{\})}}} \mylabel{empty.1} \\ \\
& \textbf{Non-terminal} \;\;\;
{\frac{G[P(A)] \; x \Lp \Tup{\Any}{\Tup{\Ex}{\List}}}
      {G[A] \; x \Lp \Tup{\Any}{\Jv((v?,\List),x,A)}}} \mylabel{var.1}
\\ \\
& \textbf{Terminal} \tenspaces 
{\frac{}{G[a] \; ax \Lp \Tup{x}{(\Nil,\{\})}}} \mylabel{char.1} \\ \\
& \tenspaces
{\frac{}{G[b] \; ax \Lp \Tup{\Fail}{\Tup{ax}{\Mon{b}}}}}, b \ne a \mylabel{char.2}
\fivespaces
{\frac{}{G[a] \; \varepsilon \Lp \Tup{\Fail}{\Tup{\varepsilon}{\Mon{a}}}}} \mylabel{char.3}
\\ \\
& \textbf{Concatenation} \;\;\;
{\frac{G[p_1] \; xy \Lp \Tup{y}{(\Ex,L_1)}
 \fivespaces G[p_2] \; y \Lp \Tup{\Any}{(\Ey,L_2)}}
      {G[p_1 \; p_2] \; xy \Lp \Tup{\Any}{\Jj{(\Ex,L_1)}{(\Ey,L_2)}}}}   \mylabel{con.1} \\ \\
&\fivespaces 
{\frac{G[p_1] \; x \Lp \Tup{\Fail}{(\Ex,L)}}
      {G[p_1 \; p_2] \; x \Lp \Tup{\Fail}{(\Ex,L)}}} \mylabel{con.2}
\\ \\
& \textbf{Ordered Choice} \;\;\;
{\frac{G[p_1] \; xy \Lp \Tup{y}{(\Ex,L)}}
      {G[p_1 \;/\; p_2] \; xy \Lp \Tup{y}{(\Ex,L)}}} \mylabel{ord.1} \\ \\
&\fivespaces
{\frac{G[p_1] \; x \Lp \Tup{\Fail}{(\Ex,L_1)} \fivespaces G[p_2] \; x \Lp \Tup{\Any}{(\Ey,L_2)}}
      {G[p_1 \;/\; p_2] \; x \Lp \Tup{\Any}{\Jj{(\Ex,L_1)}{(\Ey,L_2)}}}}  \mylabel{ord.2}
\\ \\
& \textbf{Repetition} \;\;\;
{\frac{G[p] \; x \Lp \Tup{\Fail}{(\Ex,L)}}
      {G[p*] \; x \Lp \Tup{x}{(\Ex,L)}}} \mylabel{rep.1}
\\ \\
& \fivespaces {\frac{G[p] \; xyz \Lp \Tup{yz}{(\Ex,L_1)} \fivespaces G[p*] \; yz \Lp \Tup{z}{(\Ey,L_2)}}
      {G[p*] \; xyz \Lp \Tup{z}{\Jj{(\Ex,L_1)}{(\Ey,L_2)}}}} \mylabel{rep.2}
\\ \\
& \textbf{Negative Predicate} \;\;\;
{\frac{G[p] \; x \Lp \Tup{\Fail}{(\Ex,L)}}
      {G[!p] \; x \Lp \Tup{x}{(\Nil,\{\})}}} \mylabel{not.1}
\;\;\;
{\frac{G[p] \; xy \Lp \Tup{y}{(\Ex,L)}}
      {G[!p] \; xy \Lp \Tup{\Fail}{\Tup{xy}{\Mon{!p}}}}} \mylabel{not.2}
\end{align*}
}
\caption{Semantics of PEGs with farthest failure tracking and error lists}
\label{fig:semfarthestjoin}
\end{figure}

Rules {\bf empty.1} and {\bf char.1} do not change,
and rules {\bf char.2} and {\bf char.3} just 
return a list with the terminal that
they tried to match.

Rule {\bf var.1} uses an auxiliary function \Jv, defined
as follows:
\begin{align*}
\Jv((x,L),x,A) & =  (x,\{A\}) \\
\Jv((\Nil,\{\}),x,A) & =  (\Nil,\{\}) \\
\Jv((v,L),x,A) & =  (v,L) & \mathrm{where}\ x \neq v
\end{align*}

The idea behind \Jv\ is simple: if the right-hand side
of terminal $A$ has a farthest error position that is the
same as the current position in the input, it means that
we can treat this possible error as an error while
expecting $A$ itself. This is expressed in the first
case of \Jv\ by replacing the list of expected expressions
L with just $\{A\}$. The other cases just propagate the
failure information returned by $P(A)$.

To further understand how \Jv\ works, let us consider an example,
adapted from the grammar in Figure~\ref{fig:tiny}:
\begin{align*}
\it{Factor} & \leftarrow ``(" \; \it{Exp} \; ``)"
  \; / \; \it{Digit} \; \it{Digit}* \\
\it{Digit} & \leftarrow ``0" \;/\; ``1" \;/\; ``2"
\;/\; ``3" \;/\; ``4" \;/\; ``5" \;/\; ``6" \;/\; ``7"
\;/\; ``8" \;/\; ``9"
\end{align*}

When we try to match {\it Factor} against
an input that does not match either $``("$
or {\it Digit}, such as \texttt{id}, both
alternatives of the ordered choice fail and the farthest
failure position of both matches is the same input that
we were trying to match with {\it Factor}.
So, instead of keeping
an error list that would later give us an error message like
the following one:
\begin{verbatim}
   Unexpected 'id', expecting '(' or '0' or '1' or '2' or '3' or '4' or '5'
       or '6' or '7' or '8' or '9'
\end{verbatim}

Or even an error message like this one:
\begin{verbatim}
   Unexpected 'id', expecting '(' or Digit
\end{verbatim}

We replace the error list built during the matching of
the right-hand side of {\it Factor} with just $\{ {\it Factor} \}$.
From this new list we can get the following higher-level
error message:
\begin{verbatim}
   Unexpected 'id', expecting Factor
\end{verbatim}

If the failure occurred in the middle of the first
alternative of {\it Factor} (for example, because of a missing ``)',
we would keep the original error list instead of replacing
it with $\{ {\it Factor} \}$.

Using the names of the non-terminals in the error message instead
of a list of terminals that might have started this non-terminal
might not be better if the names of the non-terminals give no
indication to the user of what was expected. It is easy
to replace \Jv\ with just $(v?,L)$, propagating the failure
information from the expansion of the non-terminal and keeping
the error message just with symbols that the user can add
to the input. Another possibility is to have two kinds of
non-terminals, and use \Jv\ for informative non-terminals and
simple propagation for the others.

In the case of concatenation, a failure in the first
expression of the concatenation means we just need to
propagate the error list, so case {\bf con.2} is
straightforward. If both expressions succeed, we might
need to merge both error lists, if the two parts
of the concatenation have the same failure position.
Rule {\bf con.1} uses the auxiliary \J\ function, 
defined below:
\begin{align*}
\J((x,L),(\Nil,\{\})) & =  (x,L) \\
\J((\Nil,\{\}),(x,L)) & =  (x,L) \\
\J((xy,L_1),(y,L_2)) & =  (y,L_2) & \mathrm{where}\ x \neq \varepsilon \\
\J((y,L_1),(xy,L_2)) & =  (y,L_1) & \mathrm{where}\ x \neq \varepsilon \\
\J((x,L_1),(x,L_2)) & =  (x,L_1 \concat L_2)
\end{align*}

The first four cases of \J\ keep the furthest
and associated set of expected expressions. In the first
two cases, \Nil means that no error has occurred,
so any error automatically becomes the farthest
error position. In the next two cases, one of the
furthest error positions is a strict suffix of the
other, so is the furthest of the two and the other
is discarded. The only remaining possibility,
expressed in the last case, is that
the two positions are identical, and we merge their
expected sets.

The rules for ordered choice and repetition use the
same \J\ function, where applicable. Finally, the rules
for the syntactic predicate $!p$ also ignore the error
information inside $p$. Rule {\bf not.2} blames the
failure on the predicate itself.

Going back to our running example (Figure~\ref{fig:tinyerror}),
our error tracking semantics will give an error list that
lets us generate the following error message:
\begin{verbatim}
    factorial.tiny:6:1: syntax error, unexpected 'until',
                        expecting ';', '=', '<', '-', '+', '/', '*'
\end{verbatim}

The operators also end up in the error list because
their lexical rules all fail in the same position as
the semicolon. This error message is similar to the
error message we get using the error tracking
combinators that we implemented on LPeg, at the end of
Section~\ref{sec:pegs}.

We might try to tweak the error tracking heuristics of
repetition and ordered choice to ignore errors that happen
in the first symbol of the input, which would let us
take out the operators from the error list in our previous
example, and give a more succinct error message:
\begin{verbatim}
    factorial.tiny:6:1: syntax error, unexpected 'until', expecting ';'
\end{verbatim}

This heuristic is not sound in the general case, though.
Suppose we replace line 6 of Figure~\ref{fig:tinyerror}
with the following line:
\begin{verbatim}
    6 n ; until (n < 1);
\end{verbatim}

The tweaked heuristic would still produce an error list
with just the semicolon, which is clearly wrong, as the
problem is now a missing operator.

It is common in PEGs to mix scanning and parsing in the
same grammar, as syntactic predicates and repetition make
lexical patterns convenient to express. But this can lead
to problems in the automatic generation of error messages
because failures are expected while recognizing a token,
and these failures related to scanning can pollute the
error list that we use to generate the error message.
 
As an example, let us consider the lexical rule {\tt THEN}
from the PEG for the Tiny language of Section~\ref{sec:pegs}:
\[
  {\tt THEN} \leftarrow {\tt then} \; \it{!IDRest \; Skip}
\]

The pattern in the right-hand side fails if any alphanumeric
character follows {\tt then} in the input, putting the
predicate in the error list. The error will be reported to
the user as an unexpected character after {\tt then}
instead of a missing {\tt then} keyword.

One solution is to split the set of non-terminals into
{\em lexical} and {\em non-lexical} non-terminals.
Non-lexical non-terminals follow rule {\bf var.1}, but
lexical terminals follow a pair of new rules:
\begin{align*}
& 
{\frac{G[P(A)] \; xy \Lp \Tup{y}{\Tup{\Ex}{\List}}}
      {G[A] \; xy \Lp \Tup{y}{(\Nil,\{\})}}} \mylabel{lvar.1}\\ \\
& 
{\frac{G[P(A)] \; x \Lp \Tup{\Fail}{\Tup{\Ex}{\List}}}
      {G[A] \; x \Lp \Tup{\Fail}{(x, \{ A \})}}} \mylabel{lvar.2}
\end{align*}

Intuitively, a lexical non-terminal reports errors just like
a terminal. Any failures that happen in a successful match
are ignored, as they are considered to be expected, and a failed
match of the whole non-terminal gets blamed on the
non-terminal itself, at its starting position in the input.

All the extensions to the semantics of PEGs that we discussed
in this section are {\em conservative}: if a PEG fails with
some subject in the original semantics, it will also fail
in all of our extensions, and if a PEG matches some subject
and leaves a suffix, it will also match and leave the same
suffix in all of our extensions. This is stated by the
following lemma, where we use the symbol $\Lp$ to represent
the regular semantics of PEGs~\cite{mascarenhas2014}, and the symbol $\Lpp$ to represent
the extended semantics of PEGs presented in Figure~\ref{fig:semfarthest}:
\begin{lemma}[Conservativeness of farthest failure tracking]
Given a PEG $G$, a parsing expression $p$ and a subject $xy$, we have that 
$\Matg{p}{xy} \Lp y$ \,iff\; $\Matg{p}{xy} \Lpp \Tup{y}{\Ex}$,
and that 
$\Matg{p}{xy} \Lp \Fail$ \,iff\; $\Matg{p}{xy} \Lpp \Tup{\Fail}{\Ex}$.
\end{lemma}
\begin{proof}
By induction on the height of the respective proof trees. 
The proof is straightforward, since the only difference
between the rules of $\Lp$ and $\Lpp$ is presence of
the farthest failure position, but this position has
no influence on whether the expression successfully
consumes part of the input or fails.
\end{proof}

A similar Lemma for the semantics with expected expression lists
of Figure~\ref{fig:semfarthestjoin} is also straightforward to
prove, where we use the symbol $\Lpl$ to represent those extended
semantics:
\begin{lemma}[Conservativeness of farthest failure tracking with expected lists]
	Given a PEG $G$, a parsing expression $p$ and a subject $xy$, we have that 
	$\Matg{p}{xy} \Lp y$ \,iff\; $\Matg{p}{xy} \Lpl \Tup{y}{(\Ex, L)}$,
	and that 
	$\Matg{p}{xy} \Lp \Fail$ \,iff\; $\Matg{p}{xy} \Lpl \Tup{\Fail}{(\Ex,L)}$.
\end{lemma}
\begin{proof}
By induction on the height of the respective proof trees. 
The proof is also straightforward, since the only difference
between the rules of $\Lp$ and $\Lpl$ is presence of
the farthest failure position and list of expected expressions, 
but this extra information has
no influence on whether the expression successfully
consumes part of the input or fails.
\end{proof}

\end{document}